\documentclass[11pt]{article}
\usepackage[utf8]{inputenc}
\usepackage[left=1in, top=1in, right=1in, bottom=1in]{geometry}
\usepackage{soul}
\usepackage{amsthm}
\usepackage{amsmath}
\usepackage{amssymb}
\usepackage{xcolor}
\usepackage{setspace}
\usepackage{booktabs} % For formal tables
\usepackage{algorithm}
\usepackage{algorithmic}
\usepackage[colorlinks=true, citecolor=blue, linkcolor=red, urlcolor=black]{hyperref}
\usepackage[switch]{lineno}
\usepackage[nocomma, short]{optidef}

\DeclareMathOperator*{\argmin}{arg\,min}

\usepackage{amsfonts}
\usepackage{mathtools}
\usepackage{bbm}
\usepackage{natbib}%\usepackage[authoryear, square]{natbib}
\usepackage{mdframed}
\mdfsetup{% backgroundcolor=grisframe,
 skipabove=4pt,
 skipbelow=4pt,
 leftmargin=0pt,
 rightmargin=0pt,
 innertopmargin=3pt,
 innerbottommargin=3pt,
 innerleftmargin=3pt,
 innerrightmargin=3pt,
 frametitleaboveskip=3pt,
 frametitlebelowskip=3pt,
 frametitlerule=true,
 splittopskip=2\topsep}
\usepackage{enumitem}

\newtheorem{definition}{Definition}
\newtheorem{theorem}{Theorem}
\newtheorem{lemma}{Lemma}
\newtheorem{proposition}{Proposition}
\newtheorem{observation}{Observation}
\newtheorem{assumption}{Assumption}

\newtheorem*{remark}{Remark}

\DeclareMathOperator*{\E}{\mathbf{E}}
\newcommand{\dd}{\mathrm{d}}

\ttfamily

\begin{document}
%\begin{titlepage} % Title page for title and abstract only.
\title{Selling an Item through Persuasion}

\author{
Zhikang Fan \\ Renmin University of China\\ \texttt{fanzhikang@ruc.edu.cn}
\and
Weiran Shen\\ Renmin University of China \\ \texttt{shenweiran@ruc.edu.cn}
}

\maketitle

\begin{abstract}
A monopolistic seller aims to sell an indivisible item to multiple potential buyers. Each buyer's valuation depends on their private type and the item's quality. The seller can observe the quality but it is unknown to buyers. This quality information is valuable to buyers, so it is beneficial for the seller to strategically design experiments that reveal information about the quality before deciding to sell the item to whom and at what price.

We study the problem of designing a revenue-maximizing mechanism that allows the seller to disclose information and sell the item. First, we recast the revelation principle to our setting, showing that the seller can focus on one-round mechanisms without loss of generality. We then formulate the mechanism design problem as an optimization problem and derive the optimal solution in closed form. The optimal mechanism includes a set of experiments and payment functions. After eliciting buyers' types, the optimal mechanism asks a buyer to buy and sets a price accordingly. The optimal information structure involves partitioning the quality space. Additionally, we show that our results can be extended to a broader class of distributions and valuation functions.
\end{abstract}
\setcounter{tocdepth}{1} 

\onehalfspacing

\newpage 

\section{Introduction}
\label{sec:introduction}

Optimal auction design stands as a critical issue in the field of economics and game theory, garnering sustained attention over the years~\citep{myerson1981optimal,armstrong2000optimal,shi2012optimal,papadimitriou2011optimal,deng2014revenue}. The most renowned contribution in this area is the seminal work by~\cite{myerson1981optimal}. Most subsequent work has adhered to the standard assumption that buyers hold private information, and the seller aims to design a mechanism with the property of incentive compatibility to elicit buyers' true information. 

However, in reality, the seller often possesses private information about the item for sale. For instance, the seller knows the true quality of the item, while buyers can only ascertain that after purchasing and using the item for some time. This scenario motivates us to study the mechanism design problem when both parties have private information. In such a context, if the seller aims to maximize revenue, the mechanism's outcome (e.g., the winner and the payment) may need to be based on the seller's private information. Consequently, the outcome itself might partially disclose the seller's private information. As rational players, buyers will update their valuation based on this information if their valuation depends on the seller's private information (such as the item's quality). This leads to a potential situation where buyers might have an incentive to disobey the allocation. To prevent such adverse effects, the seller must carefully handle how the private information is disclosed when deciding the mechanism outcome. %But it is obviously loss generality to consider a mechanism that does not use quality information.

Such information asymmetry is ubiquitous in the literature and has attracted significant research attention~\citep{akerlof1978market,bergemann2015limits,schottmuller2023optimal}. There is a line of works investigating how an information owner can strategically disclose information to benefit themselves, also known as ``information design'' or ``Bayesian persuasion''~\citep{rayo2010optimal,kamenica2011bayesian,alonso2016persuading,bergemann2016information}. Additionally, some researchers explore how a party can leverage their information advantage to make profits, referred to as ``selling information''~\citep{babaioff2012optimal,horner2016selling,bergemann2022selling}. In most existing works, the information designer or seller is typically a third party rather than the item seller. To the best of our knowledge, there is limited understanding of how to design optimal mechanisms for a seller who possesses both private information and the item for sale.

To fill this gap, we study the optimal mechanism design problem for an item seller who possesses an information advantage over buyers regarding the item's quality. Inspired by the information design literature, we adopt the ``Bayesian persuasion''  framework~\citep{kamenica2011bayesian} to model the way of revealing information. 

\subsection{Our Results}
In this paper, we begin by defining a general mechanism space that contains all possible mechanisms the seller can use. We then demonstrate that it is without loss of generality to focus on a much smaller subspace, where each mechanism involves only a single round of communication between the seller and the buyers.

Based on this result, we formulate the optimal mechanism design problem as an optimization problem. Furthermore, we fully characterize the optimal mechanism and derive a closed-form solution to this problem. In the optimal mechanism, the seller asks the buyer with the highest ``virtual value''~\citep{myerson1981optimal} to buy the item, or refrain from selling if the highest virtual value does not exceed the seller's threshold. While we allow for randomized information disclosure, the seller consistently sends deterministic messages to the buyers. The information structure features a partition of the quality space, i.e., the seller informs the buyer whether the true quality falls within a specific subset. 

Although our mechanism shares similarities with the Myerson auction, it is derived by considering how information disclosure impacts buyers' valuations. We show that the Myerson auction can be seen as a special case of ours, where the quality remains constant. Finally, we show that our results can be easily extended to a more general setting. 

\subsection{Related work}
Our research is grounded in the literature on information design. We adopt the ``Bayesian persuasion'' framework, proposed by the seminal work~\citep{kamenica2011bayesian}, to model how the seller designs information. Most subsequent works assume that only one party has private information. The most relevant work to ours in this line of work is the persuasion model with a privately informed receiver~\citep{kolotilin2017persuasion}. However, there is a substantial distinction between persuasion and selling an item through persuasion. In our context, the seller has only one item for sale, and not every buyer can buy it if they want to. Thus the results are not comparable. \cite{guo2019interval} also studied a similar setting, showing that the optimal information structure has a nested interval structure. In contrast, we show that the information interval in our optimal mechanism can also exhibit similar structures.

Our work aligns with the research on mechanism design with information revelation. \cite{esHo2007optimal} study a setting very similar to ours, except that they assume that the seller cannot observe the quality and has no reserve price. Despite these similarities, their results are quite different from ours. They show that revealing full information is optimal, whereas our information structure features a partition of the quality space. \cite{wei2022price} study a single buyer setting and derive a closed form solution under the MHR condition. In contrast, we consider a multi-buyer setting and show that our results can be extended to a more general class of distributions. Additionally, there is a series of works focusing on information disclosure under specific auction formats, such as second-price auction~\citep{bergemann2022optimal} and posted price auction~\citep{castiglioni2022signaling}. Unlike these studies, we need to jointly design both the information structure and the selling mechanism.

Another related yet significantly different problem is the selling information~\citep{liu2021optimal,bergemann2018design,chen2020selling}. In their setting, the seller sells information before the buyer takes action and can therefore charge an upfront fee. In contrast, we consider the setting where the seller reveals information and charges the payment for the item itself, resulting in significantly different analyses. Additionally, unlike the selling of information, an indivisible item can only be sold once.

\section{Preliminaries}
% \subsection{The setup}
We consider a setting where a monopolistic seller aims to sell an indivisible item to multiple potential buyers. The seller can privately observe the item's quality, denoted by $q\in Q$. There are $n$ potential buyers, represented by the set $N = \{1, 2,\dots, n \}$. Each buyer $i$ has a private type, denoted by $t_i \in T_i$, which captures buyer $i$'s personal preference. 

The item's quality $q$ is a random variable drawn from a publicly known distribution $G(q)$. The support of $q$ is defined as $[\underline{q}, \bar{q}]$. We assume that $G(q)$ is differentiable in $[\underline{q}, \bar{q}]$ and that $g(q)$ is the corresponding probability density function(PDF). Similarly, each buyer $i$'s type $t_i$ can be described by a publicly known distribution $F_i(t_i)$ with support $T_i=[\underline{t_i}, \bar{t_i}]$. Suppose that $F_i(t_i)$ has a strictly positive and continuous density function $f_i(t_i)$. Let $T$ and $T_{-i}$ denote the product space of all buyers' supports and the product space of all buyers' supports except $i$:
\begin{align*}
	% T = [\underline{t_1}, \bar{t_1}] \times \cdots \times [\underline{t_n}, \bar{t_n}].
    T=\prod_{i\in N} T_i\quad\text{and}\quad T_{-i} = \prod_{j\in N, j\neq i } T_j.
\end{align*}
% And denoted by $T_{-i}$ the product space of all buyers' supports except $i$:
% \begin{align}
% 	T_{-i} = \prod_{j\in N, j\neq i } T_j.
% \end{align}

We assume that all the random variables defined above are independent. Thus, the joint density function on $T$ is:
\begin{align*}
	f(t) = \prod_{j\in N} f_j(t_j),
\end{align*}
where $t = (t_1, \dots, t_n)\in T$ represents the buyers' type profile. For any buyer $i$, their prior belief about the types of the other buyers is given by:
\begin{align*}
	f_{-i} (t_{-i}) = \prod_{j\in N, j \neq i} f_j(t_j),
\end{align*}
where $t_{-i} = (t_1, \dots, t_{i-1}, t_{i+1}, \dots, t_n)\in T_{-i}$.

The valuation of each buyer $i$ depends on both the quality of the item $q$ and their own type $t_i$, thus there are no externalities among buyers. Formally, let $v_i(t_i, q)$ be the valuation of buyer $i$ with type $t_i$ on an item of quality $q$. Throughout this paper, we assume that $v_i(t_i, q)$ is monotone non-decreasing in $t_i$ for any $q\in Q$. For simplicity, we assume that $v_i(t_i, q)$ is linear in $t_i$ and takes the form $v_i(t_i, q) = \alpha(q) t_i$, where $\alpha(q) > 0$ for all $q$. In Section \ref{sec:generalization}, we show that our results can also be extended to a more general class of valuation functions. 
We also assume that the seller has a personal valuation on the item, denoted by $r(q)$, which can be seen as a reserve price. 

%We allows the seller communicate with buyers before allocating the item.

% \subsection{Designing Problem}
Since the seller has an information advantage on the item's quality, it is valuable for them to consider a general mechanism where they can privately communicate with buyers, possibly in multiple rounds, before determining the winner and payments. Consequently, this mechanism design problem inherently involves both information design and mechanism design. Our goal in this paper is to determine the revenue-maximizing mechanism for the seller when they can jointly design information and sell the item. We will address the following questions in the remainder of this paper:
\begin{enumerate}
	\item In what mechanism space should the seller design the optimal mechanism?
	\item What is the revenue-maximizing mechanism within this space?
\end{enumerate}

\section{Mechanism Space}\label{sec:space}
The interactions between the seller and buyers can be described as an extensive-form game. To maximize revenue, the seller can interact with buyers through various means. The set of all possible interactions forms a very large space, making it intractable to search for a revenue-maximizing one without a well-defined mechanism space. 

Previous work~\citep{babaioff2012optimal} defined a general interactive protocol for an information seller who is faced with a single buyer. Next, we extend this protocol to our setting, where an item seller interacts with multiple buyers. Therefore, the optimal mechanism we seek is the one that provides the maximum profit for the seller within the defined space. 

\begin{definition}[General interactive protocol]\label{def:protocol}
A general interactive protocol is a mechanism that induces a finite extensive-form game between the seller and buyers. The game can be fully described by a game tree, where any non-leaf node $h$ in the game tree belongs to one of the following types:
	\begin{itemize}
		\item {\bfseries{buyer node}}, where a certain buyer sends a message to the seller, each leading to a child node;
		\item {\bfseries{seller node}}, where the seller privately reveals information to a buyer. Any seller node $h$ prescribes the seller's behavior by associating each quality $q$ with a distribution over its children nodes $C(h)$. Formally, the prescription on node $h$ is a collection of distributions $\psi_h^q \in \Delta(C(h)) $, one for each $q$;
		\item {\bfseries{transfer node}}, which has only a single child node and is associated with a monetary transfer (possibly negative) $t_h$ from a certain buyer to the seller;
		\item {\bfseries{seller decision node}}, where the seller decides to which buyer he should sell the item. We allow for randomized decisions, thus the description on each seller decision node $h$ also is a collection of distributions $\psi_h^q \in \Delta(C(h))$;
		\item {\bfseries{buyer decision node}}, where a buyer decides whether to buy when the seller asks him. Once the buyer decides to buy, the game ends and goes to a leaf node. 
	\end{itemize}
$C(h)$ denotes the set of all children nodes of any node $h$. 
\end{definition}
Similar definitions also appear in the literature on selling information or services \citep{liu2021optimal,FanS23}. In line with their definitions, we allow the seller to charge and the buyer to exit the mechanism at any time. However, a key distinction in our setting is that we permit the seller to turn to another buyer if a previous sale ask is rejected, thereby incorporating multiple seller decision nodes. This flexibility arises from our multi-buyer setting, whereas most existing works focus on scenarios involving a single buyer. 

The only constraint we impose on the game tree structure is that when a buyer is asked to make a purchase decision at a seller decision node, there must be a corresponding buyer decision node that follows, ensuring the buyer's right to reject the offer. 

\begin{remark}
The difference between a seller node and a seller decision node is that each seller decision node has at most $n+1$ children\footnote{There may be fewer than $n+1$ children nodes if some buyers leave early.}, corresponding to $n$ buyers and a leaf node (indicating the seller chooses not to sell the item).
\end{remark}

Just like in the classic definition of Bayesian games, we assume the presence of another player called \emph{nature}, who selects the type profile $(t, q)$ for the buyers and seller according to $F(t)$ and $G(q)$ before they interact. The timing of the game is as follows:
\begin{enumerate}
    \item The seller commits a general interactive protocol;
    \item Nature chooses the type profile $(t, q)$ for the buyers and seller;
    \item The buyers interact with the seller according to the committed protocol.
\end{enumerate}
The protocol prescribes the seller's behavior for each $q$, and we assume that the seller has commitment power and always adheres to the protocol.

%As described above, the game is a bit like a seller playing with multiple buyers at the same time. 

\subsection{One-round mechanism}
To start, we define the concept of direct mechanisms and recast the celebrated revelation principle~\citep{myerson1979incentive,gibbard1973manipulation} in our setting.

\begin{definition}[Direct Mechanism]
A direct mechanism is a protocol represented by a tree where all nodes with depth $n$ or less are buyer nodes, where buyers report types directly to the seller. Moreover, there are no other buyer nodes in the tree, and the path from the root to any node of depth $n$ contains exactly a buyer node for each buyer. 
\end{definition}
In a game tree, any node can be uniquely determined by the path extending from the root to that node. Hence, in a direct mechanism, the seller directly gathers all buyers' types and subsequently decides to whom the item is sold and the corresponding payments. We say a direct mechanism is \emph{truthful} if it is in the best interest of each buyer to reveal their true type, even when the buyer decides not to buy the item.

\begin{lemma}[Revelation Principle~\citep{myerson1981optimal}]
For any general interactive protocol, there exists a truthful direct mechanism that extracts the same expected revenue.
\end{lemma}
The above result enables us to focus on the set of direct mechanisms. Nevertheless, direct mechanisms can still have arbitrary arrangements of other nodes. Next, we show that our setting admits a stronger version of the revelation principle that permits us to focus on one-round mechanisms. This contrasts with selling information where it is known that when engaged in information revelation, multi-round mechanisms can extract strictly higher revenue than one-round mechanisms~\citep{babaioff2012optimal}.

A one-round mechanism is a mechanism in which the path from the root to any leaf node comprises at most one seller decision node. Within such a mechanism, the seller does not ask a second buyer even if the first one refuses to buy the item. As the seller only asks one buyer, one-round mechanisms can be represented by a payment $p^t$ and an experiment $\pi^t$, both depending on buyers' reported types. We assume that the payment function for buyer $i$ only depends on the type he reports. Hence, the payment function $p:T\mapsto \mathbb{R}^n$ can be decomposed to $n$ individual payment functions: $p_i: T_i \mapsto \mathbb{R}$, for any buyer $i\in N$. We justify this assumption as being reasonable since each buyer's utility and decision only depend on his type. Thus, the information about other buyers' types obtained from the dependence of the payment function on those types is of no use.

To describe how the seller discloses information, we need to introduce the concept of \emph{experiments}. Given a signal set $\Sigma$, an experiment $\pi: Q \mapsto \Delta(\Sigma)$ is a mapping from the quality $q$ to a distribution over the signal set $\Sigma$. With a slight abuse of notation, we use $\pi(t, q)$ to denote the distribution over $\Sigma$ and $\pi(t, q, \sigma)$ to denote the probability of sending signal $\sigma$ when buyers report types $t$ and the quality is $q$. This way of revealing information is also called ``persuasion'' or ``information structure'' in the literature~\citep{kamenica2011bayesian,bergemann2018design}.

\begin{lemma}[\cite{bergemann2018design}]
It is without loss of generality to focus on the responsive experiment where the signal space has at most the cardinality of the outcome space.
\end{lemma}

Based on the above results, we can focus on the set of experiments where there exists a one-to-one correspondence between the set of signals and the set of outcomes. In our setting, there are $n+1$ possible outcomes, with $n$ of them corresponding to each buyer obtaining the item and an additional one corresponding to no buyer getting the item. Thus, we can define $\Sigma$ as follows:
\begin{align*}
    \Sigma =  \left\{ x\in \{0, 1\}^n : \sum_{i=1}^n x_i \le 1   \right\},
\end{align*}
where $x_i=1$ corresponds to the outcome where buyer $i$ obtains the item and $x_i=0$ for all $i$ corresponds to the outcome where the seller retains the item. Hence, we can interpret each signal as an indicator of the buyer that the seller will ask. From the perspective of implementation, the seller can send each buyer a message indicating whether the buyer will be asked or not, that is, the seller sends the $i$-th element to buyer $i$. For simplicity, we denote the signal with $x_i=1$ as $s_i$, and the signal with $x_i=0$ for all $i$ as $s_0$. Therefore, when the seller sends signal $s_i$, they asks buyer $i$ to buy the item. 

Now we are prepared to formally define the one-round mechanism.
\begin{definition}[One-round Mechanism]
A one-round mechanism $\mathcal{M}$ is a set $\{( \pi^t, p^t) \}_{t\in T}$, and the game proceeds as follows:
\begin{enumerate}
    \item The seller announces a mechanism $\mathcal{M}$;
    \item The quality $q$ and buyers' types $t$ are realized;
    \item Buyers privately report their types $t_i'$ to the seller;
    \item The seller sends a signal according to $\pi(t', q)$; %by asking the corresponding buyer to buy the item;
    \item The buyer who is asked to buy the item decides whether to buy or not;
    \item If the buyer decides to buy, they pay $p_i(t_i')$.
\end{enumerate}
\end{definition}
It is easy to see that a one-round mechanism can be represented as a general interactive protocol. When the seller asks the buyer to consider buying the item, we allow the buyer to reject it if they find it unprofitable. We call a mechanism \emph{obedient} if a truthful buyer always decides to buy the item whenever asked. And if a mechanism is both truthful and obedient, we say that the mechanism is \emph{incentive compatible}. 

Next, we show that it is without loss of generality to focus on the set of one-round mechanisms.

\begin{theorem}
\label{the:revelation}
For any general interactive protocol, there exists an incentive compatible, one-round mechanism that extracts the same expected revenue.
\end{theorem}

\section{Problem Formulation}
In this section, we formulate the mechanism design problem as a mathematical program. Given a one-round mechanism $\mathcal{M}= \{( \pi^t, p^t) \}_{t\in T}$, the seller's expected revenue from using this mechanism is:
\begin{align}
	% \max_{\pi, p}  \quad 
 U_0=\int_{t\in T}\int_{q\in Q} \left[\sum_{i\in N} \pi(t, q, s_i) p_i(t_i) + \left(1- \sum_{i\in N}\pi(t, q, s_i) \right)r(q) \right] f(t) g(q) \,\mathrm{d}t \mathrm{d}q. \label{eq:obj}
\end{align}

The remainder of this section is devoted to establishing constraints on $\mathcal{M}$.

\paragraph{Probability constraints.} The seller has only a single item for sale, so $\pi$ must satisfy the following constraints:
\begin{align}
\label{eq:basic}
    \sum_{i \in N} \pi(t, q, s_i) \le 1 \quad \text{and} \quad \pi(t, q, s_i) \ge 0, \forall i, \forall t, \forall  q.
\end{align}

% \begin{remark}
% Any buyer $i$ gets utility 0 when receiving a signal $s_j \neq s_i$, thus we can omit the analysis of this part in the subsequent analysis.
% \end{remark}

\paragraph{Individual rationality constraints.} Individual rationality constraints ensure that every buyer is willing to participate in the mechanism. When the seller sends signal $s_i$, any buyer other than $i$ gets utility 0. Thus, in this case, we can safely ignore the analysis for all other buyers. The expected utility of buyer $i$ with type $t_i$, if reporting truthfully and deciding to buy the item, is:
\begin{align}
\begin{aligned}
\label{eq:utility}
	U_i(t_i) = &\E_{t_{-i} \sim F_{-i}, g\sim G } [ [v_i(t_i, q) - p_i(t_i)]\pi(t, q, s_i) ] \\
	=&\int_{ T_{-i}} \int_{Q} \left[ v_i(t_i, q) - p_i(t_i) \right] \pi(t, q, s_i) g(q) f_{-i}(t_{-i}) \,\mathrm{d} q \mathrm{d} t_{-i}.
\end{aligned}
\end{align}
Obviously, buyer $i$ gets 0 if they do not participate in the mechanism. Therefore, to ensure IR constraints, we require:
\begin{align}
\label{eq:IR}
	U_i(t_i) \ge 0, \forall t_i \in T_i, i\in N.
\end{align}

\paragraph{Incentive compatibility constraints.} Recall that in our setting, incentive compatibility requires that all buyers both truthfully report their types and buy the item whenever asked. So we need two steps to ensure the IC constraints: one step for the obedience constraints and the other for the truthfulness constraints.

For the obedience constraints, we need to ensure that any buyer $i$ is willing to buy the item, assuming they truthfully reported their types earlier. A buyer cannot forcibly buy the item if the seller decides not to sell it to him. Therefore, when receiving a signal different from $s_i$, the utility of buyer $i$ can only be 0, and we only need to guarantee that buyer $i$ is willing to buy the item when receiving signal $s_i$. Upon receiving signal $s_i$, buyer $i$ with type $t_i$ will update his belief according to the Bayes rule:
\begin{align*}
	g(q|t_i, 1) = \frac{t_{-i} \in \int_{T_{-i}} \pi(t, q, s_i) g(q) f_{-i}(t_{-i}) \,\dd{t_{-i}}  }{ \int_{t_{-i} \in T_{-i}} \int_{q\in Q} \pi(t, q, s_i) g(q) f_{-i}(t_{-i}) \,\dd{q}\dd{t_{-i}}  }.
\end{align*}
Then the buyer decides to buy the item if the following is non-negative:
\begin{align*}
\int_{q\in Q} g(q|t_i, 1) v_i(t_i, q) \,\dd q - p_i(t_i) = \frac{\int_{t_{-i} \in T_{-i}} \int_{q\in Q} [ v_i(t_i, q) - p_i(t_i) ] \pi(t, q, s_i) g(q) f_{-i}(t_{-i}) \,\dd q \dd t_{-i} }{\int_{t_{-i} \in T_{-i}} \int_{q\in Q} \pi(t, q, s_i) g(q) f_{-i}(t_{-i}) \,\dd q\dd t_{-i}}.
\end{align*}
Interestingly, the numerator is exactly $U_i(t_i)$ and the above constraint is satisfied if and only if:
\begin{align}
	U_i(t_i) \ge 0, \forall t_i \in T_i, i\in N,
\end{align}
which turns out to be the same as the IR constraints.

As for the truthfulness constraints, we need to ensure that all buyers have no incentive to misreport their types. For any buyer $i$ with type $t_i$, if he misreports type $t_{i}'$ and decides to buy the item, his expected utility is:
\begin{align}
\begin{aligned}
    U_i(t_i'; t_i) = \int_{t_{-i} \in T_{-i}} \int_{q\in Q} [v_i(t_i, q) - p_i( t_i')] \cdot \pi( (t_{-i}, t_i'),q, s_i ) g(q) f_{-i}(t_{-i})\,\dd q \dd t_{-i}.
\end{aligned}
\end{align}
Note that $\pi((t_{-i}, t_i'), q, s_i)$ may not be obedient for buyer $i$. Thus, to guarantee truthfulness, we also need to ensure that truthful reporting leads to a higher expected utility than that of misreporting and not buying, which is 0. Combining the two cases, we have:
\begin{align*}
	U_i(t_i) \ge \max \{ U_i(t_i'; t_i), 0 \}, \forall t_i, t_i' \in T_i, i\in N.
\end{align*}
Note that $U_i(t_i) \ge 0$ is already satisfied by the IR constraints. Hence, the IC constraints can be simplified to:
\begin{align}
\label{eq:IC}
	U_i(t_i) \ge U_i(t_i'; t_i), \forall t_i, t_i' \in T_i,  i\in N.
\end{align}

Overall, the mechanism design problem for the seller can be formulated as the following optimization problem: 
%, with variables $\pi(t, q, s_i), p_i(t_i)$:
\begin{maxi}
{\pi, p}
{\quad \quad \quad U_0 }
{\label{eq:LP}}
{}
\addConstraint{U_i(t_i)}{\ge 0, \quad}{\forall  t_i \in T_i, i\in N}
\addConstraint{U_i(t_i)}{\ge U_i(t_i'; t_i), \quad}{\forall  t_i, t_i'\in T_i, i\in N}
\addConstraint{ \sum_{i\in N} \pi(t, q, s_i)  }{\le 1, }{\forall t\in T, q\in Q}
\addConstraint{\pi(t, q, s_i)}{\ge 0, \quad}{\forall t\in T, q\in Q, i\in N}
\end{maxi}

\section{The Optimal Mechanism}
\label{sec:optimal}
In this section, we derive an optimal solution to Program \eqref{eq:LP} in a closed form. Before presenting results, we first define some useful functions and concepts.

\begin{definition}[Virtual Value Function \citep{myerson1981optimal}]
For any random variable $w_i$ with PDF $x_i(w_i)$ and CDF $X_i(w_i)$, the virtual value function is defined as:
\begin{align*}
	\phi_i(w_i) = w_i - \frac{1- X_i(w_i)}{x_i(w_i)}.
\end{align*}
\end{definition}
% The above definition is consistent with Myerson's defintion. Next we define the regularity condition.
\begin{definition}[Regularity \citep{myerson1981optimal}]
We say a problem is regular, if the function $\phi_i(t_i)$ is monotone increasing in $t_i$, for all $i \in N$.
\end{definition}
This condition is a standard technical condition and is commonly seen in the literature. In Section \ref{sec:irregular}, we will extend our results to irregular distributions. 

Our solution falls into the following category of threshold mechanisms.
\begin{definition}[Threshold Mechanisms]
A mechanism $(\pi, p)$ is called a threshold mechanism if there exist $n$ pairs of functions $(\lambda_i,\eta_i)$, where $\lambda_i:T_i\mapsto \mathbb{R}$ and $\eta_i:T_{-i}\times Q\mapsto\mathbb{R}$, such that for any $t\in T$, $q\in Q$, and $i \in N$:
\begin{align*}
    \pi(t, q, s_i) = \begin{cases}
        1 \quad \text{if }  \lambda_i(t_i) \ge \eta_i(t_{-i},q)\\
	0 \quad \text{otherwise}
    \end{cases}.
\end{align*}
	% where $\eta = \max_{j\in N, j\neq i} \lambda_j(t_j) $. 
\end{definition}

In a threshold mechanism, its information structure $\pi$ is fully characterized by the set of function pairs $(\lambda_i,\eta_i)$. Note that the term ``threshold'' is only used to describe the information structure and does not have any restriction on payment functions. 

Now we are ready to present the optimal mechanism.
\begin{theorem}
\label{the:main}
If a problem is regular, then a threshold mechanism with the following threshold functions and payment functions forms an optimal mechanism:
\begin{gather*}
    \lambda_i(t_i) = \phi_i(t_i),\\
    \eta_i(t_{-i},q)=\max\left\{\max_{j\ne i}\{\phi_j(t_j)\},\frac{r(q)}{\alpha(q)}\right\},
\end{gather*}
\begin{align}
    p_i^*(t_i) = &\frac{1}{\int_{T_{-i}}\int_{Q} \pi^*(t, q, s_i)g(q)f_{-i}(t_{-i})\,\dd q \dd t_{-i}  }  \nonumber \\ 
	&\cdot \left[ \int_{T_{-i}}\int_{Q} v_i(t_i, q)\pi^*(t, q, s_i) g(q)f_{-i}(t_{-i})\,\dd q \dd t_{-i}  - \int_{\underline{t_i}}^{t_i} R_i^{\pi^*}(x)\,\dd x  \right].\label{eq:opt_buyer_payment}
\end{align}

\end{theorem}

In the optimal mechanism, the seller will only ask buyer $i$ to consider purchasing the item when the virtual value of buyer $i$ is the largest among all buyers and is also greater than $\frac{r(q)}{\alpha(q)}$. Therefore, whenever the seller sends signal $s_i$, buyer $i$ knows that his virtual value is greater than that of others and more importantly that the $\frac{r(q)}{\alpha(q)}$ is lower than his virtual value. With this information, buyer $i$ can update his belief about $q$ and make the purchase decisions accordingly. However, for those buyers who receive a ``not buy'' signal, they cannot conclude that their virtual value is less than the item's quality, since it is possible that other buyers have a greater virtual value than him.

While our optimal mechanism shares similarities with the Myerson auction mechanism, we would like to emphasize the fundamental distinctions in our mechanism. In our setting, the seller also possesses private information, and the seller's actions may directly or indirectly reveal this information, as the seller may utilize this private information to optimize revenue. Additionally, the seller's private information can significantly change how the buyers value the item. Consequently, the seller's decision to sell the item to a specific buyer may convey information that even dissuades the buyer from purchasing the item. To mitigate this adverse effect of information transmission, the seller must carefully design both the information structure and the pricing mechanism.

% \section{Proof of the Optimal Mechanism}
The remainder of this section is devoted to the derivation of the optimal mechanism. We first characterize the mechanism space we consider, and then re-write the revenue of a mechanism in a way that can be easily optimized.
\subsection{Characterization of Feasible Mechanisms}
We call a mechanism $(\pi, p)$ \emph{feasible} if it satisfies all the constraints in Program \eqref{eq:LP}. Before we characterize the feasible space, we define the following quantity:
\begin{align*}
	R_i^{\pi}(t_i) = \int_{t_{-i} \in T_{-i}} \int_{q\in Q} \alpha(q) \pi(t, q, s_i) g(q) f_{-i}(t_{-i}) \,\dd q \dd t_{-i}.
\end{align*}
% for any buyer $i$ and type $t_i$. 
In fact, $R_i^{\pi}(t_i)$ can be seen as the expected probability, \emph{weighted} by $\alpha(q)$, that buyer $i$ of type $t_i$ is asked to buy the item given information structure $\pi$. 

\begin{lemma}
\label{lem:property}
A mechanism $(\pi, p)$ is feasible if and only if it satisfies the following constraints, for all $i\in N$:
\begin{align}
    & R^{\pi}_i(t_i) \text{ is monotone non-decreasing in $t_i$}. \label{eq:monotone}\\
    & U_i(t_i) = U_i(\underline{t_i}) + \int_{\underline{t_i}}^{t_i} R_i^{\pi}(x) \,\dd x \label{eq:IC property} \\
    & U_i(\underline{t_i}) \ge 0 \label{eq:IR property} \\
    &\sum_{i \in N} \pi(t, q, s_i) \le 1 \quad \text{and} \quad \pi(t, q, s_i) \ge 0.\label{eq:basic property}
\end{align}
\end{lemma}

We can see that constraint \eqref{eq:monotone} is quite similar to the allocation monotonicity condition in the standard auction setting. However, the distinction lies in that in our setting, the valuation of an item not only depends on the buyer's type but also on the item's quality $q$. Specifically, the valuation function has a coefficient $\alpha(q)$.

Now, we show that the mechanism described in Theorem \ref{the:main} is feasible.
\begin{lemma}
\label{lem:feasible}
    For a regular problem, mechanism $(\pi^*, p^*)$ described in Theorem \ref{the:main} is feasible.
\end{lemma}

\subsection{Deriving the Optimal Mechanism}
With Lemma \ref{lem:property}, our goal becomes to design a mechanism that maximizes Equation \eqref{eq:obj}, subject to the constraints listed in Lemma \ref{lem:property}. Before deriving the optimal solution, we need to re-write the objective function (Equation \eqref{eq:obj}). %The following lemma illustrates the equivalent form of the objective function.
\begin{lemma}
\label{lemma:rewrite obj}
	For any feasible mechanism $(\pi, p)$, the seller's objective function \eqref{eq:obj} can be written as: 
\begin{align}
\label{eq:rewrite obj}
U_0 =& -\sum_{i\in N}U_i(\underline{t_i}) + \int_{q\in Q} r(q) g(q)\,\dd q \\
 &+\int_{t\in T} \int_{q\in Q}  \sum_{i\in N}  \alpha(q) \pi(t, q, s_i)  \bigg[t_i- \frac{1-F_i(t_i)}{f_i(t_i)} - \frac{r(q)}{\alpha(q)} \bigg]  g(q) f(t) \,\dd q \dd t \nonumber .
\end{align}
\end{lemma}
\begin{proof}
Based on Equation \eqref{eq:obj}, we can rewrite the seller's objective function as:
\begin{align}
\begin{aligned}
\label{eq:U_0}
	U_0 =& \sum_{i\in N} \int_{t\in T}\int_{q\in Q} \pi(t, q, s_i) [p_i(t_i) - v_i(t_i, q)] f(t) g(q)\,\dd t \dd q \\
	&+ \sum_{i\in N} \int_{t\in T}\int_{q\in Q} \pi(t, q, s_i) [v_i(t_i, q) - r(q)]  f(t) g(q)\,\dd t \dd q \\
	&+ \int_{t\in T}\int_{q\in Q} r(q) f(t) g(q) \,\dd t \dd q.
\end{aligned}
\end{align}
Using Equation \eqref{eq:utility} and \eqref{eq:IC property}, we know that:
\begin{align}
\begin{aligned}
\label{eq:obj-1}
	&\int_{t\in T}\int_{q\in Q}\pi(t, q, s_i) [p_i(t_i) - v_i(t_i, q)] f(t) g(q)\,\dd t\dd q\\
	=& \int_{\underline{t_i}}^{\bar{t_i}} -U_i(t_i) f_i(t_i) \,\dd t_i \\
	=& - \int_{\underline{t_i}}^{\bar{t_i}} \left(U_i(\underline{t_i}) + \int_{\underline{t_i}}^{t_i} R_i^{\pi}(x) \,\dd x \right) f_i(t_i)\,\dd t_i\\
	=& -U_i(\underline{t_i}) - \int_{\underline{t_i}}^{\bar{t_i}}\int_{\underline{t_i}}^{t_i} R_i^{\pi}(x) f_i(t_i)\,\dd x \dd t_i\\
	=& -U_i(\underline{t_i}) - \int_{\underline{t_i}}^{\bar{t_i}} (1-F_i(x))R_i^{\pi}(x)\,\dd x\\
	=& -U_i(\underline{t_i}) - \int_{t\in T}(1- F_i(t_i)) \alpha(q) \pi(t, q, s_i)g(q) f_{-i}(t_{-i}) \,\dd q\dd t.
\end{aligned}
\end{align}

By definition, we have:
\begin{align}
\begin{aligned}
\label{eq:obj-2}
	v_i(t_i, q) - r(q) = \alpha(q) t_i - r(q)
\end{aligned}	
\end{align}

Plugging Equation \eqref{eq:obj-1} and \eqref{eq:obj-2} back into Equation \eqref{eq:U_0} gives:
\begin{align*}
	U_0 =& -\sum_{i\in N}U_i(\underline{t_i}) + \int_{t\in T} \int_{q\in Q}  \sum_{i\in N}  \alpha(q) \pi(t, q, s_i)  \\
	& \cdot \left[t_i - \frac{1-F_i(t_i)}{f_i(t_i)} - \frac{r(q)}{\alpha(q) } \right]   g(q) f(t) \,\dd q \dd t + \int_{q\in Q} r(q) g(q)\,\dd q
\end{align*}
\end{proof}

To show that the mechanism described in Theorem \ref{the:main} is optimal, we need to prove that it maximizes Equation \eqref{eq:rewrite obj}. 
\begin{proof}[Proof of Theorem \ref{the:main}]
The revenue equation \eqref{eq:rewrite obj} contains 3 terms, where the second term has nothing to do with mechanism, but only with the seller's value and the distribution of $q$. Thus we only need to show that the mechanism $(\pi^*, p^* )$ maximizes the first and third term in revenue equation \eqref{eq:rewrite obj} at the same time.

By definition, $\pi^*(t, q, s_i) = 1$ if and only if:
\begin{align*}
	\phi_i(t_i) \ge \frac{r(q)}{\alpha(q)} \quad \text{and} \quad \phi_i(t_i) \ge \max_{j\ne i}\{\phi_j(t_j)\}, 
\end{align*}
%where $\eta = \max_{j\in N, j\neq i} \phi_j(t_j)$. 
which means that the third term is point-wisely optimized for all $(t, q)$ pairs, and thus this term is maximized.

For the first term, if we want to maximize the revenue, we need to minimize the term $U_i(\underline{t_i})$. According to Lemma \ref{lem:property}, any feasible mechanism should ensure $U_i(\underline{t_i})$ to be non-negative. The proof of Lemma \ref{lem:feasible} shows that the mechanism $(\pi^*, p^*)$ induces $U_i(\underline{t_i}) = 0$, which also optimizes this term.

Overall, the mechanism $(\pi^*, p^*)$ optimizes these two terms at the same time, so it is an optimal mechanism. 
\end{proof}

%Due to space limitations, we defer the proof to the appendix.
%We omit the detailed proof here but offer intuitions to provide a conceptual understanding.

%Note that constraint \eqref{eq:monotone} and \eqref{eq:basic property} are only constraints on $\pi$, and constraints on payment functions are only \eqref{eq:IC property} and \eqref{eq:IR property}. Also, we find that in Equation \eqref{eq:rewrite obj}, only the first term is related to the payment function, and the third term is only related to the quality of the item itself. Therefore, the intuition is that we select a payment function that minimizes the term $U_i(t_i)$ for all $i\in N$, under constraints \eqref{eq:IC property} and \eqref{eq:IR property}. For the second term in Equation \eqref{eq:rewrite obj}, we can maximize it by pointwise optimization, without violating constraints \eqref{eq:monotone} and \eqref{eq:basic property} if the problem is regular. 

\section{Generalizations} \label{sec:generalization}
In this section, we show that our results can be extended to a broader class of distributions and valuation functions.

\subsection{Irregular Distributions}
\label{sec:irregular}
The optimal mechanism relies crucially on the regularity condition, i.e., the monotonicity of $\phi_i(t_i)$. Without this condition, the optimal mechanism would not even be feasible as it might violate constraint \eqref{eq:monotone}. Next, we extend our results to irregular cases. 

When the problem is irregular, we need to focus on the last term of the seller's revenue \eqref{eq:rewrite obj}:
\begin{align}
\label{eq:second term}
	\int_{t\in T} \int_{q\in Q}  \sum_{i\in N}  \alpha(q) \pi(t, q, s_i) \left[\phi_i(t_i) - \frac{r(q)}{\alpha(q) } \right]   \,\dd G(q) \dd F(t).
\end{align}
In irregular problems, functions $\phi_i(t_i)$ are no longer monotone increasing. We apply the ironing technique~\citep{myerson1981optimal} to make it monotone increasing in $t_i$. We provide a formal definition of this ironing technique in the appendix. 

By applying the ironing technique on function $\phi_i(t_i)$, we obtain an ironed function, denoted by $\bar{ \phi_i}(t_i)$. Now we present the optimal mechanism for the general case, which also has a threshold structure.
\begin{theorem}
\label{the:general}
For any irregular case, the threshold mechanism with the following threshold functions and payment functions is an optimal mechanism:
\begin{gather*}
\lambda_i(t_i) = \bar{\phi_i}(t_i),\\
\eta_i(t_{-i},q)=\max\left\{\max_{j\ne i}\{\bar{\phi}_j(t_j)\},\frac{r(q)}{\alpha(q)}\right\},
\end{gather*}

\begin{align}
    p_i^*(t) = &\frac{1}{\int_{T_{-i}}\int_{Q} \pi^*(t, q, s_i)g(q)f_{-i}(t_{-i})\,\dd q \dd t_{-i}  } \nonumber \\ 
	&\cdot \left[ \int_{T_{-i}}\int_{Q} v_i(t_i, q)\pi^*(t, q, s_i) g(q)f_{-i}(t_{-i})\,\dd q \dd t_{-i}  - \int_{\underline{t_i}}^{t_i} R_i^{\pi^*}(x)\,\dd x  \right].\label{eq:general_opt_payment}
\end{align}
\end{theorem}

\begin{proof}
By definition, we have $h_i(F_i(t_i)) = \phi_i(t_i) $, $l_i(F_i(t_i)) = \bar{\psi_i}(t_i) $. So Equation \eqref{eq:second term} can be rewritten as:
\begin{align*}
	&\int_T \int_Q  \sum_{i\in N}  \alpha(q) \pi(t, q, s_i) g(q) f(t) \left[\phi_i(t_i) - \frac{r(q)}{\alpha(q)} \right]  \,\dd q \dd t \\
	=& \int_T \int_Q  \sum_{i\in N}  \alpha(q) \pi(t, q, s_i) g(q) f(t) \left[\bar{ \phi_i}(t_i) - \frac{r(q)}{\alpha(q)} \right]    \,\dd q \dd t \\
	& + \int_T \int_Q  \sum_{i\in N}  \alpha(q) \pi(t, q, s_i) g(q) f(t)  \cdot [  h_i(F_i(t_i)) - l_i(F_i(t_i))  ]      \,\dd q \dd t 
\end{align*}
We can simplify the second term above via integration by parts:
\begin{align*}
	&\int_T \int_Q  \sum_{i\in N}  \alpha(q) \pi(t, q, s_i) g(q) f(t) [  h_i(F_i(t_i)) - l_i(F_i(t_i))  ]    \,\dd q \dd t\\
	=& \int_{\underline{t_i} }^{\bar{t_i}} [h_i(F_i(t_i)) - l_i(F_i(t_i))]R_i^{\pi}(t_i) \,\dd F_i(t_i)\\
	=& [H_i(F_i(t_i)) - L_i(F_i(t_i))] R_i^{\pi}(t_i)   |_{\underline{t_i}}^{\bar{t_i}} - \int_{\underline{t_i}}^{\bar{t_i}} [H_i(F_i(t_i)) - L_i(F_i(t_i))] \,\dd R_i^{\pi}(t_i).
\end{align*}
By definition, $L_i$ is the convex full of $H_i$, then we have $L_i(0) = H_i(0)$ and $L_i(1) = H_i(1)$. Therefore, the first term is 0, and Equation \eqref{eq:second term} becomes:
\begin{align}
\label{eq:second rewrite}
	&\int_T \int_Q  \sum_{i\in N}  \alpha(q) \pi(t, q, s_i) g(q) f(t) \left[\bar{ \phi_i}(t_i)- \frac{r(q)}{\alpha(q) } \right]    \,\dd q \dd t - \int_{\underline{t_i}}^{\bar{t_i}} [H_i(F_i(t_i)) - L_i(F_i(t_i))] \,\dd R_i^{\pi}(t_i).
\end{align}

To prove the optimality, it suffices to show that the mechanism $(\pi^*, p^*)$ maximizes all the terms in Equation \eqref{eq:rewrite obj} simultaneously. 

Note that the payment function is the same as in Theorem \ref{the:main}, so according to the proof of Theorem \ref{the:main}, the first term $U_i(t_i) = 0$, implying this term is optimized.

For the third term, it can be transferred to Equation \eqref{eq:second rewrite} based on the above analysis. 

Now we show that $(\pi^*, p^*)$ maximizes Equation \eqref{eq:second rewrite}. The first term is maximized by $\pi^*$, since $\pi^*(t, q, s_i) = 1$ if and only if:
\begin{align*}
	\bar{\phi_i}(t_i) \ge \frac{r(q)}{\alpha(q) } \quad \text{and} \quad \bar{\phi_i}(t_i) \ge \max_{j\ne i}\{ \bar{\phi_j}(t_j)\}.
\end{align*}
%where $\bar{\eta} = \max_{j\in N, j\neq i} \bar{\phi_j} (t_j) $.

For the second term of Equation \eqref{eq:second rewrite}, it is worth noting that $H_i(F_i(t_i)) - L_i(F_i(t_i)) \ge 0$ since $L_i(\omega)$ is the convex hull of $H_i(\omega)$. Additionally, the ironed function is monotone increasing in $t_i$, so $\pi^*(t, q, s_i)$ is also monotone increasing in $t_i$. Therefore, $R_i^{\pi^*}(t_i)$ is monotone increasing in $t_i$. This means that $\dd R_i^{\pi^*}(t_i) $ is always non-negative. Therefore, in order to show that this term is maximized, it suffices to prove that this term is equal to 0. Actually, it is only necessary to consider the cases where $H_i(F_i(t_i)) - L_i(F_i(t_i)) >0$. In such cases, $t_i$ must fall within an ironed interval $I$, thus function $L_i(\omega)$ is linear in the interval $I$. This implies $l_i(\omega)= \bar{\phi_i}(t_i)$ is a constant and thus $R_i^{\pi^*}(t_i) $ is also a constant in interval $I$, leading to $\dd R_i^{\pi^*}(t_i) = 0$. 

The last term in Equation \eqref{eq:rewrite obj} is a constant given a problem instance. Overall, the mechanism $(\pi^*, p^*)$ maximizes all the terms in Equation \eqref{eq:rewrite obj} simultaneously, proving it to be an optimal mechanism.
\end{proof}

\subsection{General Utility Functions}
We have characterized the optimal mechanism with value functions that are linear and monotone increasing in $t_i$ for all $i\in N$. In this section, we consider a broader class of value functions and show that our results can be easily extended to this case, subject to the following assumptions.
\begin{assumption}
\label{cond:convex}
For any $i\in N$ and any $q\in Q$, when viewed as a function of $t_i$, $v_i(t_i, q)$ is convex and monotone non-decreasing.
\end{assumption}

\begin{assumption}
\label{cond:monotone}
For any $i\in N$ and any $q\in Q$, when viewed as a function of $t_i$, the following function is monotone non-decreasing:
\begin{gather*}
    \phi_i(t_i) = \frac{v_i(t_i, q)}{\frac{\partial v_i(t_i, q)}{ \partial t_i}} - \frac{1- F_i(t_i)}{f_i(t_i)}.
\end{gather*}
% $\phi_i(t_i) = \frac{v_i(t_i, q)}{v'_i(t_i, q)} - \frac{1- F_i(t_i)}{f_i(t_i)}$ is monotone non-decreasing in $t_i$, for all $i\in N$, where $v'_i(t_i, q) = \frac{\partial v_i(t_i, q)}{ \partial t_i}$.
\end{assumption}
Based on Assumption \ref{cond:convex}, we can also provide a characterization of feasible mechanisms, similar to Lemma \ref{lem:property}, but with the following re-defined $R_i^{\pi}(t_i)$:
\begin{align*}
	R_i^{\pi}(t_i) = \int_{T_{-i}} \int_{Q} \frac{\partial v_i(t_i, q) }{ \partial t_i} \pi(t, q, s_i) g(q) f_{-i}(t_{-i}) \,\dd q \dd t_{-i}.
\end{align*}

The convexity condition in Assumption \ref{cond:convex} is crucial for ensuring truthfulness. The techniques used in Lemma \ref{lem:property} can be applied in a similar way, allowing us to omit the proof for brevity. Based on the characterization, we can rewrite the revenue of the seller, similar to Lemma \ref{lemma:rewrite obj}. Here we present the term that needs to be maximized point-wise:
\begin{align*}
    v_i(t_i, q) - \frac{\partial v_i(t_i, q) }{ \partial t_i} \frac{1- F(t_i)}{f(t_i)} - r(q).
\end{align*}
% \begin{multline*}
%     \int_T \int_Q \sum_{i\in N} \pi(t, q, s_i) \bigg[ v_i(t_i, q) \\- \frac{\partial v_i(t_i, q) }{ \partial t_i} \frac{1- F(t_i)}{f(t_i)} - r(q) \bigg]  \,\dd G(q) \dd F(t).
% \end{multline*}
The main purpose of introducing Assumption \ref{cond:monotone} is to avoid the necessity of using the ironing procedure when maximizing the above equation. Thus we can maximize the revenue equation  point-wise by a threshold mechanism, without violating the monotonicity constraints on the re-defined quantity $R_i^{\pi}(t_i)$. 

\section{Discussions}\label{sec:dis}
In this section, we briefly discuss the optimal mechanism under different information discrimination settings.
\subsection{Information Discrimination}
The optimal mechanism features information discrimination: the winner will be aware that his virtual value is greater than $\frac{r(q)}{\alpha(q)}$. It is worth considering what the optimal mechanism would be when information discrimination is not permitted. We define the following two degrees of information discrimination.

% Before presenting the results, we first need to define the degree of information discrimination. 
%It is well-known that three different degrees of price discrimination sets different prices based on the buyer's identity or purchasing quantity. 
% However, to our best knowledge, there doesn't seem to exist a corresponding definition to describe the degree of information discrimination. 

\begin{definition}[First-degree information discrimination]
	The seller sends different signals to different buyers.
\end{definition}
%{\color{red} not sure which is first-degree }
\begin{definition}[Second-degree information discrimination]
	The seller uses different experiments based on different buyer types and sends different signals to different buyers. 
\end{definition}

The optimal mechanism employs second-degree information discrimination, i.e., the experiments $\pi$ is a function of buyer types $t$, and different buyers receive different messages since each buyer receives an element of an $n$-dimensional vector in $\Sigma$. However, we can clearly send the same message to all buyers by sending the entire $n$-dimensional vector. This crucially depends on two important characteristics of our setting:
\begin{itemize}
    \item Each buyer's valuation does not depend on other buyers' types. This means the information about other buyers' types does not benefit any buyer.
    \item In the optimal mechanism, the seller asks at most one buyer. This is correlated with the elements in any outcome vector.
\end{itemize}
Hence, we have the following proposition.
\begin{proposition}
If second-degree discrimination is allowed but the first-degree discrimination is not allowed, the mechanism described in Theorem \ref{the:main} remains optimal.
\end{proposition}

However, if the seller is also unable to use different experiments based on different buyer types, there is no need to gather the buyers' private information before deciding on the experiment. %It turns out that the seller can only use a pre-defined experiment. 
In this case, we obtain the following result.
\begin{proposition}
\label{prop:constant price}
If second-degree discrimination is not allowed, the optimal mechanism must charge a constant price.
\end{proposition}

%\subsection{Proof of Proposition \ref{prop:constant price} }

\subsection{Information Structure}
Throughout this paper, we don't assume the monotonicity of function $\frac{r(q) }{\alpha(q) }$. In this section, we will discuss how the monotonicity of the function $\frac{r(q) }{\alpha(q) }$ affects the information threshold. For brevity, let $\xi(q) = \frac{r(q)}{\alpha(q) }$.

In optimal mechanism, once buyer $i$ with type $t_i$ receives signal 1, he will have a belief that $\phi_i(t_i) \ge \xi(q) $, but what about the belief in $q$? Based on this intuition, we have the following observation.

\begin{observation}
Let $\xi^{-1}(q)$ be the inverse function of $\xi(q)$.
\begin{itemize}
    \item If $\xi(q)$ is monotone increasing in $q$, buyer $i$'s belief over $q$ is a ``lower'' interval $[\underline{q}, \xi^{-1}(\phi_i(t_i))]$.
    \item If $\xi(q)$ is monotone decreasing in $q$, buyer $i$'s belief over $q$ is a ``upper'' interval $ [\xi^{-1}( \phi_i(t_i) ),  \bar{q} ]$.
    \item Otherwise, buyer $i$'s belief over $q$ is cut into multiple segments.
\end{itemize}
\end{observation}
The information structure is related to which party values the quality of the item more. As the quality increases, if the buyer's valuation increases at a faster rate, the seller is inclined to offer the high-quality item to the buyer. In contrast, if the seller's reserve price increases faster, the seller may opt to sell lower-quality items to the buyer.

\subsection{Connection with Myerson Auction}
The crucial distinction between our setting and the standard auction setting lies in the fact that in our case, the seller owns private information regarding the item's quality $q$ and thus can jointly design information and payment functions. Moreover, if the seller has no information advantage over $q$, they cannot reveal information about $q$, and the allocation outcome will be independent of $q$. That is, function $\pi$ is independent of $q$, then $\pi^*$ becomes 
\begin{gather*}
    \pi^*(t, s_i) = 
    \begin{cases}
        1 & \text{if } \phi_i(t_i) \ge \max_{j \neq i}\{\phi_j(t_j), \bar{r} \}\\
        0 & \text{otherwise}
    \end{cases},
\end{gather*}
where $\bar{r}$ represents the seller's valuation of the item. It can be seen as a constant when $q$ is public information. 

We can observe that $\pi^*$ precisely aligns with Myerson's allocation function. Hence, the Myerson auction can be regarded as a special case of ours, where the quality $q$ remains constant.

\section{Conclusion}
In this paper, we studied the optimal mechanism design problem for a seller who has an information advantage over the item's quality. The optimal mechanism has a threshold structure: asks one buyer to buy or not sell the item to any buyer. The information structure involves partitioning the quality space and the number of partitions depends on the monotonicity of function $\frac{r(q)}{\alpha(q)}$. The Myerson auction can be seen as a special case of our mechanism: when the quality of item $q$ privately observed by the seller becomes public information, our optimal mechanism reduces to the Myerson auction.

\bibliographystyle{apalike}
\bibliography{bib}

%\newpage

% Appendix
\appendix

\section*{Appendix}
\section{Omitted Proofs in Section \ref{sec:space}}
\subsection{Proof of Theorem \ref{the:revelation}}
\begin{proof}
We construct a one-round mechanism and show that the constructed mechanism gives all buyers the same expected utilities and extracts the same expected revenue for the seller. 

Let $Z$ be the set of all leaf nodes. Each leaf node $z\in Z$ corresponds to an outcome, and these outcomes have varying utility for different buyers. From the perspective of buyer $i$, we represent $Z_i$ as all possible outcomes for buyer $i$. In a game tree, any node can be uniquely determined by the path from the root to the node. With a slight abuse of notation, we use $z\in Z_i$ to denote both a node and its path.

Then we define the utility of buyer $i$ when reaching leaf node $z \in Z_i$:
\begin{align*}
	u_i(z) = \begin{cases}
		v_i(t_i, q) - \tau(z)  &\text{if buyer $i$ gets the item}\\
		-\tau(z)  &\text{otherwise}
	\end{cases},
\end{align*}
where $\tau(z)$ is the total payment transfer from buyer $i$ to the seller along the path $z$. Any buyer gets the item if and only if the seller recommends him to buy and the buyer is willing to buy.

% \begin{remark}
Note that we include the behavior of nature in path $z$, since we view nature as a player. Thus, $q$ can be inferred from $z$.
% $q$ can be inferred from $z_i$ since $z_i$ is a path that records the behavior of everyone along the path, including the behavior of nature.
% \end{remark}

According to definition \ref{def:protocol}, any general interactive protocol prescribes the seller's behavior by associating the seller node $h$ with a distribution $\psi_h$ and the transfer node $h$ with transfer $t_h$. Let $\beta_i(t_i, h, c)$ be buyer $i$'s optimal strategy in the above mechanism, i.e., $\beta_i(t_i, h, c)$ is the probability of a buyer with type $t_i$ choosing child node $c$ at buyer node $h$. Let $\beta = (\beta_i)_{i\in N}$ be the strategy profile. Let $Z_i(t, q)$ be the set of all possible leaf nodes $z$ that buyer $i$ may reach and $\rho_i(z) $ be the probability of reaching node $z$ when we start from the root $(t, q)$ and use $\psi_h, t_h, \beta$ to move down the tree. According to the signal that buyer $i$ receives, the set $Z_i(t, q)$ can be divided into $\{Z_i(t, q, 1) \}$, $\{Z_i(t, q, 0) \}$ and $Z_i(t, q, \emptyset)$, where $Z_i(t, q, \emptyset)$ means that buyer $i$ exits the mechanism before the seller decision node.

We can denote by $Z_i^+(t, q, 1)$ be the set of leaf nodes $z$ where the seller asks buyer $i$ to buy and the outcome is buyer $i$ gets the item, and $Z_i^-(t, q, 1)$ be the set of lead nodes where buyer $i$ refuses to buy the item. Thus, $Z_i(t, q) = Z_i^+(t, q, 1) \cup Z_i^-(t, q, 1) \cup Z(t, q, 0) \cup Z(t, q, \emptyset) $ .

In one-round mechanisms, we can without loss of generality assume that any buyer $i$ will always be willing to buy the item if the seller sends signal $s_i$, since otherwise, the seller can simply re-label the signal $s_i$ as $s_0$, without affecting the outcome.

We construct the following one-round mechanism $\mathcal{M}$:
\begin{enumerate}
	\item Let $\Sigma = \{s_i \}_{i\in N \cup \{0\} } $ be the signal set, and set $\pi(t, q, s_i) = \sum_{z\in Z_i^+(t, q, 1)} \rho_i(z) $ to be the probability of sending signal $s_i$;
	\item When receiving signal $s_i$, buyer $i$ pays $p_i(t_i)$ to buy the item, where:
	\begin{align*}
		p_i(t_i) = \frac{ \mathbb{E}_{t_{-i}\sim F_{-i}, q\sim G } \left[\sum_{z \in Z_i(t, q) } \rho_i(z) \tau(z) \right]  }{ \mathbb{E}_{t_{-i} \sim F_{-i}, q\sim G } [\pi(t, q, s_i) ]}.
	\end{align*}
\end{enumerate}

Now we show that the constructed mechanism brings the same expected utility to all buyers as the original one. To illustrate, we will focus on discussing the buyer $i$, the rest of buyers are the same.

In the original mechanism, we have $\sum_{z\in Z_i(t, q)} \rho_i(z) = 1$. The expected utility of buyer $i$ is:
\begin{align*}
	&\mathbb{E}_{t_{-i}\sim F_{-i}, q\sim G} \left[ \sum_{z\in Z_i(t, q)} \rho_i(z) u_i(z) \right] \\
	=& \mathbb{E}_{t_{-i}\sim F_{-i}, q\sim G} \left[ \sum_{z\in Z_i^+(t, q, 1)} \rho_i(z) u_i(z) + \right.\\
	&\left. +  \sum_{z\in Z_i(t, q) - Z_i^+(t, q, 1)  } \rho_i(z) u_i(z)  \right] \\
	=& \mathbb{E}_{t_{-i}, q} \left[ \sum_{z\in Z_i^+(t, q, 1)} \rho_i(z) v_i(t_i, q) - \sum_{z\in Z_i(t, q)} \rho_i(z) \tau(z)    \right].
\end{align*}

In the constructed mechanism, the signal set $\Sigma =  \{s_i \}_{i\in N \cup \{0\} } $. When receiving signal 0, the buyer $i$ can do nothing and the expected utility is 0. Let $g(q|t_i, 1)$ be the posterior belief of buyer $i$ over quality $q$ upon receiving signal $1$. We have:
\begin{align*}
	g(q|t_i, 1) = \frac{\int_{T_{-i}} \pi(t, q, s_i) g(q) f_{-i}(t_{-i}) \,\dd q \dd t_{-i}  }{ \mathbb{E}_{t_{-i} \sim F_{-i}, q\sim G } [\pi(t, q, s_i) ] }.
\end{align*}
Therefore, buyer $i$'s expected utility is:
\begin{align*}
	&\mathbb{E}_{q\sim g(q|t_i, 1)} [v_i(t_i, q)] - p_i(t_i) \\
	=& \frac{\mathbb{E}_{t_{-i} \sim F_{-i}, q\sim G } [\pi(t, q, s_i) v_i(t_i, q) ] }{\mathbb{E}_{t_{-i} \sim F_{-i}, q\sim G } [\pi(t, q, s_i) ] }.
\end{align*}
So the expected utility of buyer $i$ with type $t_i$ in the constructed mechanism is:
\begin{align*}
	&\pi(s_i| t_i) \left[ \mathbb{E}_{q\sim g(q|t_i, s_i)} [v_i(t_i, q)] - p_i(t_i) \right] \\
	=& \mathbb{E}_{t_{-i} \sim F_{-i}, q\sim G } [\pi(t, q, s_i) v_i(t_i, q) ]  \\
	& \quad \quad - \mathbb{E}_{t_{-i}\sim F_{-i}, q\sim G } [\sum_{z \in Z_i(t, q)} \rho_i(z) \tau(z) ]  \\
	=& \mathbb{E}_{t_{-i}, q } \left[ \sum_{z\in Z_i^+(t, q, 1) } \rho_i(z)v_i(t_i, q) - \sum_{z\in Z_i(t, q)} \rho_i(z) \tau(z)   \right], 
\end{align*}
where $\pi(s_i|t_i)$ is the probability that buyer $i$ with type $t_i$ receives signal $s_i$, and the last equation is due to the construction of $\pi(t, q, s_i)$.

The above analysis shows that the constructed mechanism brings the same expected utility for buyer $i$ as the original mechanism. The same result applies to other buyers. Therefore, the constructed mechanism extracts the same expected revenue for the seller as the original mechanism.
\end{proof}

\section{Omitted Proofs in Section \ref{sec:optimal}}
\subsection{Proof of Lemma \ref{lem:property} }
\begin{proof} We first show that the conditions in Lemma \ref{lem:property} are necessary for any feasible mechanism. 

We can rewrite $U_i(t_i'; t_i)$ as follows:
\begin{align*}
	U_i(t_i'; t_i) =& \int_{T_{-i}}\int_{Q} [v_i(t_i', q) - p_i( t_i') + v_i(t_i, q) \\ & - v_i(t_i', q) ] \pi( (t_{-i}, t_i'), q, s_i) g(q) f_{-i}(t_{-i}) \,\dd q \dd t_{-i}\\
	=& \ U_i(t_i') + (t_i - t_i') R_i^{\pi}(t_i').
\end{align*}
Then the IC constraint \eqref{eq:IC} is equivalent to:
\begin{align}
\label{eq:left}
	U_i(t_i) \ge U_i(t_i') + (t_i - t_i') R_i^{\pi}(t_i').
\end{align}
The above equation holds for any type $t_i$ and $t_i'$. So switching $t_i$ and $t_i'$, we obtain:
\begin{align}
\label{eq:right}
	U_i(t_i') \ge U_i(t_i) + (t_i' - t_i) R_i^{\pi}(t_i).
\end{align}
Combining Equation $\eqref{eq:left}$ and $\eqref{eq:right}$ gives:
\begin{align}
	(t_i - t_i')R_i^{\pi}(t_i') \le U_i(t_i) - U_i(t_i') \le (t_i - t_i') R_i^{\pi}(t_i).
\end{align}
Notice that $(t_i - t_i') (R_i^{\pi}(t_i) - R_i^{\pi}(t_i')) \ge 0 $, which implies constraint \eqref{eq:monotone}.

When $t_i > t_i'$, we can divide the above inequality by $t_i - t_i'$:
\begin{align*}
	R_i^{\pi}(t_i') \le \frac{U_i(t_i) - U_i(t_i')}{t_i - t_i'} \le R_i^{\pi}(t_i).
\end{align*}
Letting $t_i \rightarrow t_i'$, we have:
\begin{align}
\label{eq:inte}
	\frac{\dd U_i(t_i)}{\dd t_i} = R_i^{\pi}(t_i).
\end{align}
The above equation still holds if $t_i < t_i'$. Therefore, Equation \eqref{eq:IC property} follows.

Next, we show that IR constraint \eqref{eq:IR} implies constraint \eqref{eq:IR property}. By definition, $R_i^{\pi}(t_i)$ is non-negative. Together with Equation \eqref{eq:inte} and \eqref{eq:monotone}, we know that $U_i(t_i)$ is convex and take its minimum value at $\underline{t_i}$. Therefore, to satisfy IR constraint \eqref{eq:IR}, we only need to ensure $U_i(\underline{t_i}) \ge 0$, that is Equation \eqref{eq:IR property}.

Now we need to show that the conditions in Lemma \ref{lem:property} is also sufficient. 

Constraint \eqref{eq:basic property} directly follows from \eqref{eq:basic}. The IC constraint \eqref{eq:IC} is equivalent to:
\begin{align*}
	U_i(t_i) \ge U_i(t_i') + (t_i - t_i') R_i^{\pi}(t_i'),
\end{align*}
which is implied by constraint \eqref{eq:IC property} since when $t_i > t_i'$, we have:
\begin{align*}
	U_i(t_i) - U_i(t_i') = \int_{t_i'}^{t_i} R_i^{\pi}(x) \dd x &\ge  \int_{t_i'}^{t_i} R_i^{\pi}(t_i') \dd x\\
	&= (t_i - t_i') R_i^{\pi}(t_i')
\end{align*}
Similarly, when $t_i' > t_i$, we also have $U_i(t_i) - U_i(t_i') \ge (t_i - t_i') R_i^{\pi}(t_i')$.

Recall that the IR constraint is $U_i(t_i) \ge 0$. By definition $R_i^{\pi}(t_i) \ge 0$, together with condition \eqref{eq:IC property}, we know that $U_i(t_i)$ is monotone increasing in $t_i$. Thus, $U_i(\underline{t_i}) \ge 0$ implies $U_i(t_i)\ge 0$ for all $t_i \in T_i$.
\end{proof}

\subsection{Proof of Lemma \ref{lem:feasible} }
\begin{proof}
Observe that if a problem is regular, the information structure $\pi^*(t, q, s_i)$ is monotone non-decreasing in $t_i$, for all $i \in N$. This means that $R_i^{\pi}(t_i)$ is non-decreasing in $t_i$, satisfying constraint \eqref{eq:monotone}.
	
Based on Equation \eqref{eq:utility} and payment function $p^*$, the utility of any buyer $i$ of type $t_i$ is equal to:
\begin{align*}
	U_i(t_i) =& \int_{T_{-i}} \int_{Q} \left[ v_i(t_i, q) - p_i^*(t_i) \right] \pi(t, q, s_i) g(q) f_{-i}(t_{-i}) \,\dd q \dd t_{-i} \\
	=& \int_{\underline{t_i}}^{t_i} R_i^{\pi^*}(x) \,\dd x.
\end{align*}
This means that $U_i(\underline{t_i}) = 0$ and that
\begin{align*}
	\frac{\dd U_i(t_i) }{\dd t_i} = R_i^{\pi^*}(x),
\end{align*}
satisfying constraints \eqref{eq:IC property} and \eqref{eq:IR property}.
\end{proof}

\subsection{Formal Definition of the Ironing}
\begin{definition}[Ironing]
Let $\psi_i(x)$ be any non-monotone function.

\begin{enumerate}
	\item Define a random variable $\omega = F_i(x)$ and let
	\begin{align*}
		h_i(\omega) = \psi_i(F^{-1}_{i}(\omega)  )
	\end{align*}
	where $F^{-1}_{i}(\omega)$ is the inverse function of $F_i(x)$. Note that $F_i(x)$ is continuous and strictly increasing since we assume that the density function $f_i(x)$ is strictly positive. Thus the inverse function $F^{-1}_i$ is also continuous and increasing.
	\item Let $H_i: [0, 1] \mapsto \mathbf{R} $ be the integral of $h_i(\omega)$:
	\begin{align*}
		H_i(\omega) = \int_0^w h_i(r) \,\dd r.
	\end{align*}
	\item Let $L_i: [0, 1] \mapsto \mathbf{R}$ be the convex hull of the function $H_i$:
	\begin{align*}
		L_i(\omega) = \min \{ \epsilon H_i(\omega_1) + (1- \epsilon) H_i(\omega_2) \},
	\end{align*}
	where $\epsilon, \omega_1, \omega_2 \in [0, 1]$ and $\epsilon \omega_1 + (1- \epsilon) \omega_2 = \omega$.
	\item Define $l_i: [0, 1] \mapsto \mathbf{R}$ such that:
	\begin{align*}
		l_i(\omega) = L_i'(\omega).
	\end{align*}
	\item Then we obtain the ironed function $\bar{ \psi_i} $:
	\begin{align*}
		\bar{\psi_i} (x) = l_i(\omega) = l_i(F_i(x)).
	\end{align*}
\end{enumerate}
\end{definition}

\section{Omitted Proofs in Section \ref{sec:dis}}
\subsection{Proof of Proposition \ref{prop:constant price} }
\begin{proof}
    %The intuition is straightforward. 
    If the seller uses the same experiment and sends the same signal to all buyers, then no matter what type the buyers report, they get the same amount of information. So to ensure IC, the seller needs to set a constant price, otherwise they will report a type $t_i' = \argmin_{t_i} p_i^*(t_i)$.
\end{proof}

\end{document}